\documentclass{article}
\usepackage{latex8}

\usepackage{epsfig}
\usepackage{amssymb}
\usepackage{amsmath}
\usepackage{theorem}

\usepackage{pstricks}

\usepackage{msc}

\usepackage{algorithm}
\usepackage{algorithmic}

\usepackage[all]{xy}

\newcommand{\ports}{{\mathcal P}}
\newcommand{\Act}{{\mathcal Act}}

\newcommand{\wconf}{\sqsubseteq_w}
\newcommand{\sconf}{\sqsubseteq_s}

\newcommand{\tr}{{\mathcal L}}
\newcommand{\trpre}{{\mathcal L}_{pc}}
\newcommand{\trmax}{{\mathcal {\tilde L}}}

\newcommand{\makeFSM}{{\mathcal F}}

\newcommand{\error}{s_{e}}

\newcommand{\pref}{pref}

\newtheorem{theorem}{Theorem}
\newtheorem{lemma}{Lemma}
\newtheorem{proposition}{Proposition}
\newtheorem{definition}{Definition}

\newenvironment{proof}
{\noindent {\bf Proof} \\} {\hfill$\Box$}

\begin{document}


\title{Checking Finite State Machine Conformance when there are Distributed Observations}

\author{Robert M. Hierons}

\maketitle

\begin{abstract}
This paper concerns state-based systems that interact with their environment at physically distributed interfaces, called ports.
When such a system is used a projection of the global trace, called a local trace, is observed at each port.
This leads to the environment having reduced observational power:
the set of local traces observed need not uniquely define the global trace
that occurred.
We consider the previously defined implementation relation $\sconf$ and start by
investigating the problem of defining a language
$\trmax (M)$ for a multi-port finite state machine (FSM) $M$ such that
$N \sconf M$ if and only if every global trace of $N$ is in $\trmax (M)$.
The motivation is that
if we can produce such a language $\trmax (M)$ then this can potentially be used to inform development
and testing.
We show that $\trmax (M)$ can be uniquely defined but need not be regular.
We then prove that it is generally undecidable whether
$N \sconf M$,
a consequence of this result being that it is undecidable
whether there is a test case that is capable of distinguishing two states
or two multi-port FSM in distributed testing.
This result complements a previous result that it is undecidable
whether there is a test case that is guaranteed to distinguish two states or multi-port FSMs.
We also give some conditions under which $N \sconf M$ is decidable.
We then consider the implementation relation $\sconf^k$ that only concerns input
sequences of length $k$ or less.
Naturally, 
given FSMs $N$ and $M$ it is decidable whether $N \sconf^k M$
since only a finite set of traces is relevant.
We prove that if we place bounds on $k$ and the number of ports
then we can decide $N \sconf^k M$ in polynomial time
but otherwise this problem is NP-hard.
\end{abstract}
%
%
%
%
%
%
%
%

\section{Introduction}

Many systems interact with their environment at multiple physically distributed interfaces,
called ports,
with
web-services, cloud systems and wireless sensor networks being important
classes of such systems.
When we test a system that has multiple ports we place a local tester
at each port and the local tester at port $p$ only observes the events at $p$. 
This has led to the (ISO standardised) definition of the distributed test architecture in
which we have a set of distributed testers,
the testers do not communicate with one another during testing,
and there is no global clock \cite{iso9646_1}.
While it is sometimes possible to make testing more effective
by allowing the testers to exchange coordination messages
during testing \cite{cacciari99,RafiqC03},
this is not always feasible and the distributed test architecture
is typically simpler and cheaper to implement.
Importantly,
the situation in which separate agents (users or testers) interact with the system at
its ports can correspond to the expected use of the system.

Distributed systems often have a persistent internal state
and such systems are thus modeled or specified using
state-based languages.
In the context of testing the focus has largely been on
finite state machines (FSMs) and input output transition
systems (IOTSs).
This is both because such approaches are suitable and
because most tools and techniques for model-based testing\footnote{In model-based testing,
test automation is based on a model of the expected behaviour of the system or some
aspect of this expected behaviour.}
transform the models,
written in a high-level notation,
to an FSM or IOTS \cite{CartaxoMN11,Grieskamp06,Grieskamp11,FHP02,Tretmans08}.
Model-based testing has received much recent attention
since it facilitates test automation,
the results of a recent major industrial project showing
the potential for significant reductions in the cost of testing \cite{Grieskamp11}.

This paper concerns problems related to developing a multi-port system
based on a (multi-port) FSM model/specification.
Much of the work in the area of distributed testing has focussed on FSM models
\cite{AlmeidaMSTV10,dssouli85,dssouli86,Khoumsi02,sarikaya84,UralW06},
although there has also been work that considers more general models
such as IOTSs and variants of IOTSs
\cite{HaarJJ07,HJJ08,HieTestCom08,ATVA08,HieronsN10}.
While IOTSs are more expressive,
this paper explores decidability and complexity issues in distributed testing
and so we restrict attention to finite state models and,
in particular, 
to multi-port FSMs.
Naturally,
the negative decidability and complexity results proved in this paper
extend immediately to IOTSs.

When a state-based system interacts with its environment there is
a resultant sequence of inputs and outputs called a \emph{global trace}.
When there are physically distributed ports the user or tester
at a port $p$ only observes the sequence of events that occur at $p$,
the projection at $p$ of the global trace,
and this is called a \emph{local trace}.
It is known that the local testers only observing local traces introduces additional issues
into testing \cite{AlmeidaMSTV10,dssouli85,dssouli86,HaarJJ07,HJJ08,Khoumsi02,sarikaya84,UralW06}.

Previous work has shown that distributed testing introduces additional controllability
and observability problems.
A controllability problem occurs when a tester does not know when to supply
an input due to it not observing the events at the other ports
\cite{sarikaya84,dssouli85}.
Consider, for example,
the global trace shown in Figure \ref{fig:cont}.
We  use diagrams (message sequence charts)
such as this to represent scenarios.
In such diagrams vertical lines represent processes
and time progresses as we go down a line.
In this case the system under test (SUT) has two ports,
$1$ and $2$,
we have one vertical line representing the SUT,
one representing the local tester at port $1$,
and one representing the local tester at port $2$.
There is a controllability problem because the tester at port $2$
should send input $x'$ after $y$ has been sent by the SUT
but cannot know when this has happened since it does not
observe the events at port $1$

\begin{figure}
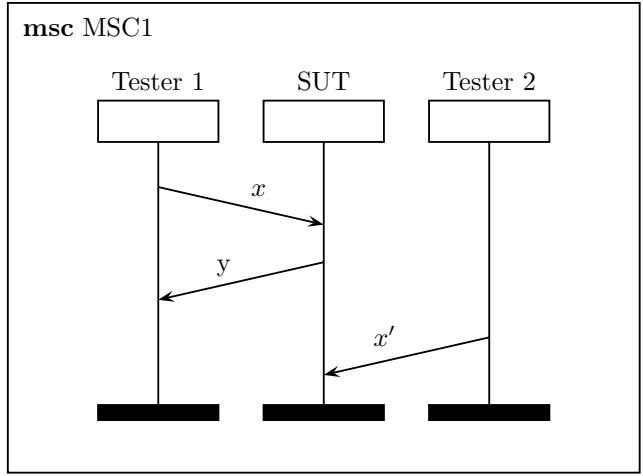

\begin{center}
\begin{msc}{MSC1}
\declinst{t1}{Tester $1$}{}
\declinst{sut}{SUT}{}
\declinst{t2}{Tester $2$}{}

\mess{$x$}{t1}{sut}[1]
\nextlevel
\nextlevel
\mess{y}{sut}{t1}[1]
\nextlevel
\nextlevel
\mess{$x'$}{t2}{sut}[1]
\nextlevel

\end{msc}
\caption{A controllability problem caused by input $x'$}
\label{fig:cont}
\end{center}
\end{figure}

Observability problems refer to the fact that the observational ability
of a set of distributed testers is less than that of a global tester since
the set of local traces need not uniquely define
the global trace that occurred  \cite{dssouli86}.
Consider, for example,
the global traces $\sigma$ and $\sigma'$ shown in Figures \ref{fig:obs1}
and \ref{fig:obs2} respectively.
These global traces are different but the local testers observe the
same local traces:
in each case the tester at port $1$ observes $xyxy$ and the tester at
port $2$ observes $y'$.
Recent work has defined new notions of conformance (implementation relations)
that recognise this reduced observational power of the environment \cite{Hierons_distrib_Oracle,HieTestCom08}.
These implementation relations essentially say that the SUT conforms to the
specification if the environment cannot observe a failure.
When using such implementation relations, we do not have to consider
observability problems:
if a global trace $\sigma$ of the SUT is observationally equivalent to one in the
specification then $\sigma$ is considered to be an allowed behaviour
since a set of distributed testers or users would not observe a failure.

\begin{figure}
\begin{center}
\begin{msc}{MSC2}
\declinst{t1}{Tester $1$}{}
\declinst{sut}{SUT}{}
\declinst{t2}{Tester $2$}{}
\mess{$x$}{t1}{sut}[1]
\nextlevel
\nextlevel
\mess{$y$}{sut}{t1}[1]
\mess{$y'$}{sut}{t2}[1]
\nextlevel
\nextlevel
\mess{$x$}{t1}{sut}[1]
\nextlevel
\nextlevel
\mess{$y$}{sut}{t1}[1]
\nextlevel
\end{msc}
\caption{Global trace $\sigma$.} \label{fig:obs1}
\end{center}
\end{figure}

\begin{figure}
\begin{center}
\begin{msc}{MSC3}
\declinst{t1}{Port $1$}{}
\declinst{sut}{SUT}{}
\declinst{t2}{Port $2$}{}
\mess{$x$}{t1}{sut}[1]
\nextlevel
\nextlevel
\mess{$y$}{sut}{t1}[1]
\nextlevel
\nextlevel
\mess{$x$}{t1}{sut}[1]
\nextlevel
\nextlevel
\mess{$y$}{sut}{t1}[1]
\mess{$y'$}{sut}{t2}[1]
\nextlevel
\end{msc}
\caption{Global trace $\sigma'$.} \label{fig:obs2}
\end{center}
\end{figure}

Given multi-port FSMs $N$ and $M$,
there are two notions of conformance for situations in which distributed observations are made:
weak conformance ($\wconf$) and strong conformance ($\sconf$).
Under $\wconf$,
it is sufficient that for every global trace $\sigma$ of $N$ and port $p$
there is some global trace $\sigma_p$ of $M$ such that
$\sigma$ and $\sigma_p$ are indistinguishable at port $p$;
they have the same local traces at $p$.
In contrast, under $\sconf$ we require that 
for every global trace $\sigma$ of $N$ 
there is some global trace $\sigma'$ of $M$ such that
$\sigma$ and $\sigma'$ are indistinguishable at all of the ports.
To see the difference, let us suppose that there are two allowed
responses to input $x_1$ at port $1$:  either $y_1$ at port $1$ and $y_2$ at port
$2$ (forming global trace $\sigma$) or $y'_1$ at port $1$ and $y'_2$ at port $2$
 (forming global trace $\sigma'$) .
Under $\wconf$ it is acceptable for the SUT to respond to $x_1$ with
$y_1$ at port $1$ and $y'_2$ at port $2$ since the local trace
at port $1$ is $x_1y_1$,
which is a projection of $\sigma$,
and the local trace at port $2$ is $y'_2$, which is a projection of $\sigma'$.
However, this is not acceptable under $\sconf$ since there is no
global trace of the specification that has projection $x_1y_1$ at port $1$
and projection $y'_2$ at port $2$.

One of the benefits of using an FSM to model the required behaviour
of a system that interacts with its environment at only one port
is that there are standard algorithms for many problems that are relevant
to test generation.
For example, 
we can decide whether there are strategies (test cases) that reach or distinguish
states \cite{alur95}
and such strategies are used by many test generation algorithms
\cite{aho91,chow78,hennie64,luo94a,PetrenkoY05,ural97}.
In addition, if we have an FSM specification $M$ and an FSM
design $N$ then we can decide whether $N$ conforms to $M$.
Thus,
if we wish to adapt standard FSM test techniques to the situation
where we have distributed testers then we need to investigate
corresponding problems for multi-port FSMs.
Recent work has shown that it is undecidable whether there
is a strategy that is guaranteed to reach a state or distinguishes two states
of an FSM in distributed testing  \cite{Hierons10_SICOMP}.
However,
this left open the question of whether one can decide
whether one FSM conforms to another.
It also left open the related question of
whether it is decidable whether there is a strategy that is capable of
distinguishing two
FSMs\footnote{It is decidable whether there is a strategy that
is capable of reaching a given state of an FSM.}.

This paper concerns the problem of deciding,
for multi-port FSMs $M$ and $N$,
whether $N$ conforms to $M$.
Clearly,
this can be decided in low order polynomial time for $\wconf$:
for each port $p$ we simply compare the projections of $N$ and $M$ at $p$.
However, $\wconf$ will often be too weak since it assumes that we can never
have the situation in which an agent is aware of observations
made at two or more ports.
We therefore focus on the implementation relation $\sconf$.

We start by investigating the question of whether,
given a multi-port FSM $M$,
we can define a language $\trmax (M)$ such that
for every multi-port FSM $N$ we have that
$N \sconf M$ if and only if all global traces of $N$ are in $\trmax (M)$.
If we can define such an $\trmax (M)$ then there is the potential to explore
properties of this in order to find algorithms for deciding $N \sconf M$
for classes of $N$ and $M$.
There is also the potential to base testing and development on $\trmax (M)$.
It has already
been shown that we can produce such a language $\trmax (M)$
for the special case where we restrict testing to
controllable input sequences and are testing from deterministic FSMs \cite{Hierons10}.
We prove that $\trmax (M)$ is uniquely defined but need not be regular.

We then consider the problem of determining whether $N \sconf M$
for multi-port FSMs $N$ an $M$ and prove that this is generally undecidable.
We also give some conditions under which $N \sconf M$
is decidable.
Clearly,
this problem is important when we are checking an FSM design
against an FSM specification.
In addition, $N \sconf M$ if no \emph{possible} behaviour of $N$ can be
distinguished from the behaviours of $M$.
Thus,
since it is undecidable whether $N \sconf M$
it is also undecidable whether there is a test case that is \emph{capable} of
distinguishing two states or FSMs.
This complements the result that it is undecidable whether there is a
test case that is guaranteed to distinguish two states or FSMs  \cite{Hierons10_SICOMP}.
However,
the proofs use very different approaches:
the proof of the previous result  \cite{Hierons10_SICOMP} used results from multi-player games
while in this paper we develop and then use results regarding
multi-tape automata.
Note that many traditional methods for testing from an FSM use sequences that distinguish
between states,
in order to check that a (prefix of a) test case takes the SUT to a correct state
\cite{aho91,chow78,hennie64,luo94a,PetrenkoY05,ural97}.
The results in this paper and in \cite{Hierons10_SICOMP}
suggest that it will be difficult to adapt such techniques for distributed testing.
In addition,
we can represent a possible fault in the SUT by an FSM $N$
formed by introducing the fault into the specification FSM $M$:
the results in this paper show that it is undecidable whether
there is a test case that can detect such a `fault',
and thus also whether it represents an incorrect implementation.

Since it is undecidable whether $N \sconf M$,
we define a weaker implementation relation $\sconf^k$ that only
considers sequences of length $k$ or less.
This is relevant when we know a bound on the length of
sequences in use or we know that the system will be
reset after at most $k$ inputs have been received.
For example,
a protocol might have a bound on the number of
steps that can occur before a `disconnect' happens.
Naturally, it is decidable whether $N \sconf^k M$ since we
only have to reason about finite sets of global traces.
We prove that if we place a bound on $k$ and the number of ports
then we can decide whether $N \sconf^k M$ in polynomial time
but the problem is NP-hard if we do not have such bounds.

This paper is structured as follows.
Section \ref{section:preliminaries} provides preliminary definitions.
Section \ref{section:models} then investigates the problem of defining
a language $\trmax (M)$ such that
for every multi-port FSM $N$ we have that
$N \sconf M$ if and only if all global traces of $N$ are in $\trmax (M)$.
In Section \ref{section:multitape}
we prove results regarding multi-tape automata that we use in
Section \ref{section:decide_strong}
to prove that it is generally undecidable whether $N \sconf M$.
Section \ref{section:decide_strong}
also gives conditions under which $N \sconf M$ is decidable.
In Section \ref{section:bounded}
we then explore $ \sconf^k$.
Finally, 
in Section \ref{section:conclusions}
we draw conclusions and discuss possible lines of future work.

\section{Preliminaries}\label{section:preliminaries}

This paper concerns the testing of state-based systems whose behaviour
is characterised by the input/output sequences (global traces) that they can produce.
Given a set $A$ we let $A^*$ denote the set of sequences formed
from elements of $A$ and we let $\epsilon$ denote the empty sequence.
In addition, $A^+$ denotes the set of non-empty sequences in $A^*$.
Given sequence $\sigma \in A^*$ we let $\pref(\sigma)$ denote
the set of prefixes of $\sigma$.
We are interested in finite state machines,
which define global traces (input/output sequences).
Given a global trace $\sigma = x_1/y_1 \ldots x_k/y_k$,
in which $x_1, \ldots, x_k$ are inputs and $y_1, \ldots, y_k$ are outputs,
the prefixes of $\sigma$ are the global traces of the form
$x_1/y_1 \ldots x_j/y_j$ with $j \leq k$.

In this paper we investigate the situation in which a system interacts with its
environment at $n$ physically distributed interfaces, called ports.
We let $\ports = \{1, \ldots, n\}$ denote the names of these ports.
Then a multi-port FSM $M$ is defined by a tuple
$(S,s_0,I,O,h)$ in which $S$ is the finite set of states,
$s_0 \in S$ is the initial state, $I$ is the finite input alphabet,
$O$ is the finite output alphabet, and $h$ is the transition relation.
The set of inputs is partitioned into subsets $I_1, \ldots, I_n$
such that for $p \in \ports$ we have that $I_p$ is the set of inputs
that can be received at port $p$.
Similarly, for port $p$ we let $O_p$ denote the set of output that can be observed at $p$.
As is usual \cite{dssouli85,dssouli86,sarikaya84,UralW06}
we allow an input to lead to outputs at several ports and so we
let $O = ((O_1 \cup \{-\}) \times \ldots \times (O_n \cup \{-\}))$
in which $-$ denotes null output.
We can ensure that the $I_p$ and also the $O_p$ are pairwise
disjoint by labelling input and output with the port name,
where necessary.
We let $\Act = I \cup O$ denote the set of possible observations
and for $p \in \ports$ we let $\Act_p = I_p \cup O_p$ denote the set of
possible observations at port $p$.

The transition relation $h$ is of type $S \times I \leftrightarrow S \times O$ and
should be interpreted in the following way:
if $(s',y) \in h(s,x)$, $y = (z_1, \ldots, z_n)$,
and $M$ receives input $x$ when in state
$s$ then it can move to state $s'$ and send output $z_p$ to port $p$ (all $p \in \ports$).
This defines the \emph{transition} $(s,s',x/y)$,
which is a \emph{self-loop transition} if $s=s'$.
Since we only consider multi-port FSMs in this paper,
we simply call them FSMs.
The FSM $M$ is said to be a \emph{deterministic FSM (DFSM)} if 
$|h(s,x)| \leq 1$ for all $s \in S$ and $x \in I$.

An FSM $M$ is completely-specified if for every state $s$ and input $x$,
we have that $h(s,x) \neq \emptyset$.
A sequence $(s_1,s_2,x_1/y_1)(s_2,s_3,x_2/y_2) \ldots (s_{k},s_{k+1}k,x_{k}/y_{k})$
of consecutive transitions is said to be a \emph{path},
which has \emph{starting state} $s_1$ and \emph{ending state} $s_{k+1}$.
This path has \emph{label} $x_1/y_1 \ldots x_k/y_k$, which is called a (global) \emph{trace}.
Further,
$x_1 \ldots x_k$ and $y_1 \ldots y_k$ are said to be the \emph{input portion} and
the \emph{output portion} respectively of $x_1/y_1 \ldots x_k/y_k$.
A path is a \emph{cycle} if its starting and ending states are the same.
The FSM $M$ defines the regular language $L(M)$ of the labels of paths of
$M$ that have starting state $s_0$.
Given state $s \in S$ of $M$ we let $L_M(s)$ denote the set of global traces
that are labels of paths of $M$ with starting state $s$,
and so $L(M) = L_M(s_0)$.
We say that $M$ is \emph{initially connected} if for every state $s$ of $M$
there is a path that has starting state $s_0$ and ending state $s$.
Throughout this paper we assume that any FSM considered is completely-specified
and initially connected.
Where this condition does not hold we can remove the states that cannot be reached
and we can complete the FSM by, for example,
either adding self-loop transitions with null output or transitions to an error state.

At times we will use results regarding finite automata (FA) and
so we briefly define FA here.
A FA $M$ is defined by a tuple
$(S,s_0,X,h,F)$ in which $S$ is the finite set of states,
$s_0 \in S$ is the initial state, $X$ is the finite alphabet,
$h$ is the transition relation,
and $F \subseteq S$ is the set of final states.
The transition relation has type $S \times (X \cup \{\tau\}) \times S$ where
$\tau$ represents a silent transition that is not observed.
The notions of a path and its label,
which does not include instances of $\tau$,
correspond to those defined for FSMs and so are not defined here.
The FA $M$ defines the language $L(M)$ of labels
of paths that have starting state $s_0$ and an ending
state in $F$.

For a global trace $\sigma$ and port $p \in \ports$ let $\pi_p(\sigma)$ denote
the \emph{local trace} formed by removing from $\sigma$ all elements that do
not occur at $p$.
This is defined by the following rules in which $\sigma$ is a global
trace (see, for example, \cite{HieronsU08}).

\begin{eqnarray*}
\pi_p(\epsilon) & = & \epsilon \\
\pi_p((x/(z_1, \ldots, z_n))\sigma) & = & \pi_p(\sigma) \mbox{ if } x \not \in I_p \wedge z_p = - \\
\pi_p((x/(z_1, \ldots, z_n))\sigma) & = & x\pi_p(\sigma) \mbox{ if } x \in I_p \wedge z_p = - \\
\pi_p((x/(z_1, \ldots, z_n))\sigma) & = & z_p\pi_p(\sigma)
\mbox{ if } x \not \in X_p \wedge z_p \neq - \\
\pi_p((x/(z_1, \ldots, z_n))\sigma) & = & xz_p\pi_p(\sigma)
\mbox{ if } x \in X_p \wedge z_p \neq -
\end{eqnarray*}

Given a set $A$ of global traces and port $p$ we let $\pi_p(A)$
denote the set of projections of sequences in $A$.
Thus, $\pi_p(A) = \{ \pi_p(\sigma) | \sigma \in A\}$.

In the distributed test architecture,
a local tester at port $p \in \ports$ only observes events from $\Act_p$.
Thus, two global traces $\sigma$ and $\sigma'$ are indistinguishable if they have the
same projections at every port and we denote this $\sigma \sim \sigma'$.
More formally, we say that $\sigma \sim \sigma'$
if for all $p \in \ports$ we have that $\pi_p(\sigma) = \pi_p(\sigma')$.

Given an FSM $M$, we let $\tr(M)$ denote the set of global sequences that
are equivalent to elements of $L(M)$ under $\sim$.
These are the sequences that are indistinguishable from sequences in $L(M)$
when distributed observations are made.
Previous work has defined two conformance relations for testing
from an FSM that reflect the observational power of distributed testing \cite{Hierons_distrib_Oracle}.
In some situations the agents at the separate ports of the SUT will never
interact with one another or share information with other agents
that can interact.
In such cases it is sufficient for
the local trace observed at a port $p$ to be a local trace of $M$.
This situation is captured by the following conformance relation.

\begin{definition}
Given FSMs $N$ and $M$ with the same input and output alphabets and the same
set of ports, $N \wconf M$ if for every global trace $\sigma \in L(N)$ and port $p$
there exists some $\sigma_p \in L(M)$ such that $\pi_p(\sigma_p) =
\pi_p(\sigma)$.  $N$ is then said to \emph{weakly conform} to $M$.
\end{definition}

However, sometimes there is the potential for information from separate testers to be received
by an external agent. 
For example, there may be a central controller that receives the observations
made by each tester and thus knows the projection of the global trace at
each port.
This leads to the following stronger conformance relation.

\begin{definition}
Given FSMs $N$ and $M$ with the same input and output alphabets and the same
set of ports, $N \sconf M$ if for every global trace $\sigma \in L(N)$ 
there exists some $\sigma' \in L(M)$ such that $\sigma' \sim \sigma$.  $N$ is then said to \emph{strongly conform} to  $M$.
\end{definition}

It is straightforward to see that 
given FSMs $N$ and $M$ we have that $N \sconf M$ if and only if $L(N) \subseteq \tr(M)$.
It is also clear that $N \sconf M$ implies that $N \wconf M$.
In order to see that $\sconf$ is strictly stronger than $\wconf$ it is sufficient to consider
the FSMs $M_1$ and $N_1$ shown in Figure \ref{fig:m1n1}.
Clearly we do not have that $N_1 \sconf M_1$ since $M_1$ has no global trace equivalent
to $x_1/(y_1,y'_2)$ under $\sim$.
However, for every global trace $\sigma$ of $N_1$ and port $p$ there is a global trace $\sigma'$ of $M_1$
such that $\pi_p(\sigma) = \pi_p(\sigma')$.
Thus, 
we have that $N_1 \wconf M_1$.

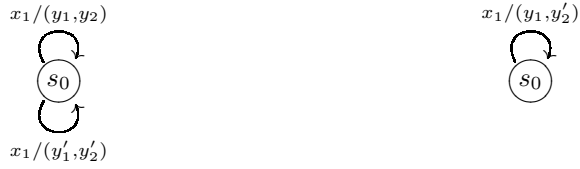
\begin{figure}
\begin{center}
\[\UseTips
\entrymodifiers = {}
\xymatrix @+1cm {
*++[o][F]{s_0} \ar@(ul,ur)[]^{x_1/(y_1,y_2)} \ar@(dl,dr)[]_{x_1/(y'_1,y'_2)} & & & *++[o][F]{s_0} \ar@(ul,ur)[]^{x_1/(y_1,y'_2)} \\
}
\]
\caption{Finite State Machines $M_1$ and $N_1$} \label{fig:m1n1}
\end{center}
\end{figure}

\section{Test models for conformance}\label{section:models}

In this section we investigate the problem of defining a language $\trmax (M)$
for an FSM $M$ such that $N \sconf M$ if and only if $L(N) \subseteq
\trmax (M)$. 
The motivation is that if we are developing the SUT from $M$
and we do have some such $\trmax (M)$ then we can use standard approaches to refine $\trmax (M)$.
In addition, if in testing we wish to test for $\sconf$ but we can
connect the local testers to form a global tester then we should compare the global traces of the SUT with $\trmax (M)$:
if we compare the global traces from the SUT with $L(M)$ then we could lead to the SUT $N$ being declared faulty
even if $N \sconf M$. 
It would be particularly useful if we  could find an FSM or IOTS $M'$ such that $L(M')
= \trmax (M)$;
we could then test $N$ against $M'$ using normal test methods
and the usual conformance relation (trace inclusion). 

We start by considering the language $\tr (M)$.
Clearly,
we have that $N \sconf M$ if and only if $L(N) \subseteq
\trmax (M)$. 
However,
if we can represent $\tr (M)$ using an FSM or IOTS then for every
$\sigma \in \tr (M)$ and $\sigma' \in pre(\sigma)$ we must have
that $\sigma' \in \tr (M)$: $\tr (M)$ must be prefix closed.

\begin{proposition}\label{prop:not_prefix_closed}
The language $\tr (M)$ need not be prefix closed.
\end{proposition}

\begin{proof}
Consider a DFSM $M$ such that $x_1/(y_1,y_2) x_1/(y_1,-) \in
L(M)$. Then $\tr (M)$ contains  $x_1/(y_1,-)
x_1/(y_1,y_2)$ but $x_1/(y_1,-) \not \in \tr (M)$ and so $\tr(M)$ is not prefix closed.
\end{proof}

The languages defined by FSMs and IOTSs are prefix closed and so we know
that $\tr (M)$ cannot always be represented by such a model.
However,
the languages defined by finite automata need not be prefix closed.
Thus, Proposition \ref{prop:not_prefix_closed} does not preclude the possibility of
representing the $\tr (M)$ using finite automata, however, the
following does.
It is already known from Mazurkiewicz trace theory that,
where some elements of an alphabet commute (i.e. $ab = ba$)
the set of sequences equivalent to those defined by
a FA need not be regular (see, for example, \cite{diekert97}). 
It is straightforward to show that this also holds for FSMs.

\begin{proposition}\label{prop:lang_M_not_regular}
Given FSM $M$,
the language $\tr (M)$ need not be regular.
\end{proposition}

\begin{proof}
Consider the FSM $M_4$ shown in Figure \ref{fig:not_regular} and
let $L = \tr (M_4)$. Proof by contradiction: assume that $L$ is
a regular language. 

Let $L' $ be the set of global traces in which all outputs are $(-,-)$.
Clearly $L'$ is a regular language. Thus, since $L$ is
regular we must have that $L'' = L \cap L'$ is regular.
The language $L''$ is the set of global traces from $\tr(M_4)$ that
have null output.
Thus, $L''$ is the set of all global traces with inputs drawn from $\{x_1,x_2\}$ and that have
null output and in which the number of
instances of $x_2$ is either equal to the number of instances of $x_1$
or is one less than this.
However, this language is not regular, providing a contradiction as required.
\end{proof}

\begin{figure}
\begin{center}
\[\UseTips
\entrymodifiers = {}
\xymatrix @+1cm {
*++[o][F]{s_0} \ar@/^1.5pc/[rrr]^{x_1/(-,-)} \ar[d]^{x_2/(y'_1,y'_2)} & & & *++[o][F]{s_1}  
\ar@/^1.5pc/[lll]^{x_2/(-,-)} \ar[d]^{x_1/(y_1,y_2)} \\
*++[o][F]{s_3} \ar@(dl,ul)[]^{x_1/(y'_1,y'_2)} \ar@(dr,ur)[]_{x_2/(y'_1,y'_2)} & & &
*++[o][F]{s_4} \ar@(dl,ul)[]^{x_1/(y_1,y_2)} \ar@(dr,ur)[]_{x_2/(y_1,y_2)} \\
}
\]
\caption{DFSM $M_4$}
\label{fig:not_regular}
\end{center}
\end{figure}
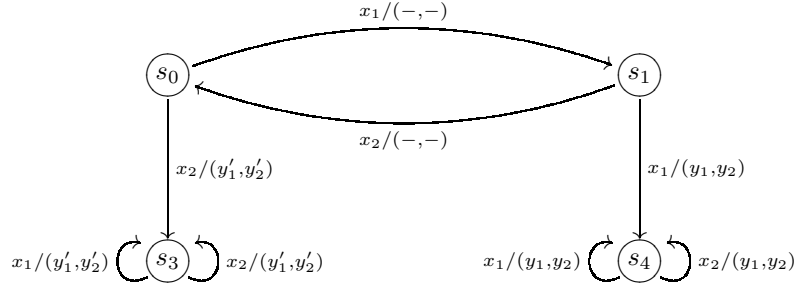

We have that $\tr (M)$ has the property we want ($N \sconf M$ if and only if
$L(N) \subseteq \tr (M)$) but that $\tr (M)$ need not be regular.
The observation that $\tr (M)$ is not prefix closed tells us
that it can contain strings that are in no $L(N)$ for an FSM or IOTS $N$
such that $N \sconf M$.
It seems natural to remove these global traces.

\begin{definition}
Given FSM $M$ let $\trpre (M)$ denote the set of global traces from
$\tr (M)$ whose prefixes are also in $\tr (M)$. More formally,
$\trpre (M) = \{\sigma \in \tr (M) | pre(\sigma) \subseteq \tr (M)\}$.
\end{definition}

\begin{proposition}
Given FSMs $M$ and $N$ we have that $N \sconf M$ if and only
if $L(N) \subseteq \trpre (M)$.
\end{proposition}

\begin{proof}
In forming $\trpre (M)$ we only remove from $\tr (M)$
a global trace $\sigma$ if some of its prefixes are
not in $\tr (M)$, and so $\sigma$ cannot be a global trace of an FSM
$N$ such that $N \sconf M$. The result therefore follows from
the fact that $N \sconf M$ if and only if $L(N) \subseteq \tr(M)$.
\end{proof}

Clearly, $\trpre (M) \subseteq \tr (M)$ and so it is natural
to ask whether there remain any global traces in $\trpre (M)$ that
cannot appear in FSMs that conform to $M$ under $\sconf$.

\begin{proposition}\label{prop:DFSM_with_sigma}
Given an FSM $M$ and global trace $\sigma \in \trpre (M)$ there exists
an FSM $N$ such that $N \sconf M$ and $\sigma \in L(N)$.
Further, this is true even if we restrict $N$ to being deterministic.
\end{proposition}

\begin{proof}
Let $\sigma$ be a global trace of length $k$ and so $\sigma = x_1/y_1 \ldots
x_k/y_k$ for some $x_1, \ldots, x_k \in I$ and $y_1, \ldots, y_k
\in O$. For all $1 \leq i \leq k$ let $\sigma_i$ denote the prefix of $\sigma$ with
length $i$.
We will construct an FSM $N'$ in the following way.
First, define an initial state $s'_0$ and for every $1 \leq i < k$ we define a
state $s'_i$ and add the transition $(s'_{i-1},s'_i,x_i/y_i)$.

Since $\sigma \in \trpre (M)$,
for each state $s'_i$, $1 \leq i \leq k$,
we choose a (not necessarily unique) state of $M$,
which we call $s_i$,
that is reached from the initial state of $M$ using a global trace $\sigma'_i \sim \sigma_i$.
We add a copy of $M$ to the structure already defined and will use
transitions to the states of this copy of $M$ in order to complete
$N'$.

First, we add the transition $(s'_{k-1},s_k,x_k/y_k)$ so if we follow
$\sigma$ by further input in $N'$ then we will obtain $\sigma$ followed
by a global trace $\sigma'$ from $s_k$ in $M$.
For all $1 \leq i < k$, $x \in I \setminus \{ x_{i} \}$, and $(s',y) \in h(s_{i-1},x)$,
we add the transition $(s'_{i-1},s',x/y)$.
Every global trace in $N'$ is either a global trace
of $M$ or is  $\sigma_i \sigma'$ for a $\sigma_i$ (which is in $\tr (M)$) and a global trace
$\sigma'$ such that $\sigma' \in L_{M}(s_i)$ and so $\sigma \in \tr (M)$. 
Further, it is clear that there is some DFSM $N$ such that $L(N) \subseteq L(N')$
and $\sigma \in L(N)$.
Thus, $N$ is a DFSM with $L(N) \subseteq L(N') \subseteq \tr (M)$ and so we have that $N'  \sconf  M$ as required.
\end{proof}

Thus, $\trpre (M)$ is the smallest language such that $N \sconf M$ if and only if $L(N) \subseteq \trpre (M)$ and so
it appears to be the language we want.

Figure \ref{fig:lang_sufficient} shows two FSMs $M$ and $M' $ such that $M'= \trpre (M)$.
Since $M'= \trpre (M)$,
for an FSM $N$ we have that $N \sconf M$ if and only
if $L(N) \subseteq L(M')$. 
This works because the only global traces that are in $\tr (M)$ but not $L(M)$
are those in the form 
$(x_1/(y_1,-))^*x_2/(-,y'_2)((x_1/(y_1,-)) + x_2/(-,y'_2))^*$
and these are included in $L(M')$.

\begin{figure}
\begin{center}
\[\UseTips
\entrymodifiers = {}
\xymatrix @+0cm {
& *++[o][F]{s_1} \ar@(dl,dr)[]_{x_2/(-,y_2)} \ar@(ul,ur)[]^{x_1/(y_1,-)} &&& & *++[o][F]{s_1} \ar@(dl,dr)[]_{x_2/(-,y_2)} \ar@(ul,ur)[]^{x_1/(y_1,-)} \\
*++[o][F]{s_0} \ar[ru]^{x_1/(y_1,-)} \ar[rd]_{x_2/(-,y'_2)} &&&&*++[o][F]{s_0} \ar@(ul,dl)[]_{x_1/(y_1,-)} \ar[ru]^{x_1/(y_1,-)} \ar[rd]_{x_2/(-,y'_2)} \\
& *++[o][F]{s_2}  \ar@(dl,dr)[]_{x_2/(-,y'_2)} \ar@(ul,ur)[]^{x_1/(y_1,-)} &&& & *++[o][F]{s_2}  \ar@(dl,dr)[]_{x_2/(-,y'_2)} \ar@(ul,ur)[]^{x_1/(y_1,-)} \\ 
}
\]
\caption{DFSMs $M$ and $M'$}
\label{fig:lang_sufficient}
\end{center}
\end{figure}
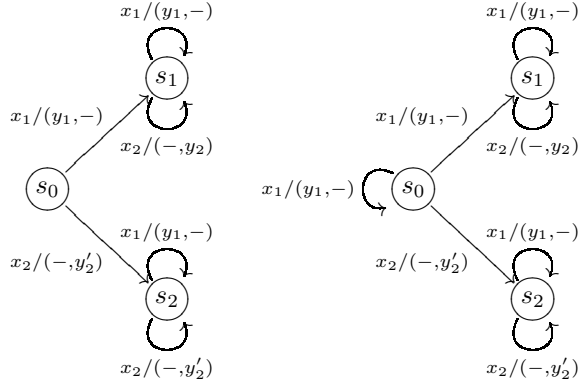

We have shown that we are required to keep all of the sequences in
$\trpre(M)$.
Now we  show that if $L = L(M')$ for some FSM $M'$ and we have that
$\trpre(M) \subset L(M')$ then the language $L(M')$ is too large.
We do this by proving that there can be a reduction $N$ of $M'$
that does not conform to $M$ under $\sconf$ and this is the case even if we
restrict $N$ to be deterministic.

\begin{proposition}\label{prop:langpre_M_cannot_add}
Given an FSM $M$, if $L' = L(M')$ for some FSM $M'$ and $\trpre(M) \subset
L'$ then there is an FSM $N$ such that $L(N) \subseteq L'$ but we do not have
that $N \sconf M$. Further, this result holds even if we restrict $N$ to being deterministic.
\end{proposition}

\begin{proof}
Since $\trpre(M) \subset L'$ there is some global trace $\sigma \in
L(M') \setminus \trpre(M)$.
Since $L' = L(M')$ for an FSM $M'$,
it is clear that there exists a DFSM $N$
such that $L(N) \subseteq L'$ and $\sigma \in L(N)$.

It is now sufficient to observe that $\sigma$ and all of its
prefixes are in $L(N)$ and since $\sigma \not \in \trpre(M)$
we must have that at least one of these sequences is not in $\tr (M)$.
\end{proof}

Thus,
the language we are looking for must contain
$\trpre(M)$ and if we restrict attention to languages defined by FSMs,
then the language cannot contain any additional global traces\footnote{We can easily
extend the proofs to more general formalisms such as IOTSs.}.  
Unfortunately, however, $\trpre (M)$ need not be regular.

\begin{proposition}\label{prop:langpre_M_not_regular}
The language $\trpre (M)$ need not be regular.
\end{proposition}

\begin{proof}
We will use the FSM $M_5$ shown in Figure \ref{fig:pre_not_regular}.
Proof by contradiction: assume that $\trpre (M_5)$ is regular.

Let $L'$ denote the regular language $(x_2/(-,-))^*(x_1/(-,-))^*$.
Consider the language $\tr (M_5) \cap L'$, which is the set of sequences of the form
$(x_2/(-,-))^*(x_1/(-,-))^*$ where the number of instances of $x_2$ can
be at worst one less than the number of instances of $x_1$.
This is prefix closed since each element of $\tr (M_5) \cap L'$
starts with all of its instances of $x_2/(-,-)$.
Clearly $\tr (M_5) \cap L'$ is not regular.

Since $\tr (M_5) \cap L'$  is prefix closed, $\trpre(M_5) \cap L' = \tr (M_5) \cap L'$.
As a result, we know that $\trpre(M_5) \cap L' $ is not regular.
But $L'$ is regular and so if $\trpre(M_5)$ was regular then $\trpre(M_5) \cap L'$ would also be regular.
This provides a contradiction as required.
\end{proof}

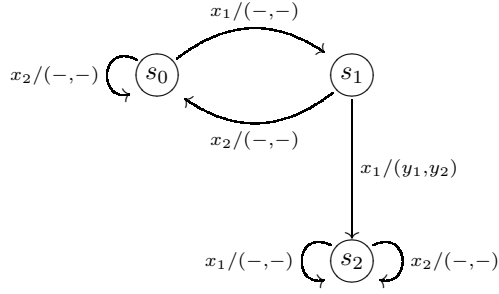
\begin{figure}
\begin{center}
\[\UseTips
\entrymodifiers = {}
\xymatrix @+1cm {
*++[o][F]{s_0} \ar@(ul,dl)[]_{x_2/(-,-)} \ar@/^1.5pc/[r]^{x_1/(-,-)} & *++[o][F]{s_1} \ar@/^1.5pc/[l]^{x_2/(-,-)} \ar[d]^{x_1/(y_1,y_2)} \\
& *++[o][F]{s_2}   \ar@(ul,dl)[]_{x_1/(-,-)}  \ar@(ur,dr)[]^{x_2/(-,-)} \\
}
\]
\caption{DFSM $M_5$}
\label{fig:pre_not_regular}
\end{center}
\end{figure}

We therefore obtain the following result.

\begin{theorem}
Given an FSM $M$ there need not exist an FSM $M'$
with the property that for all FSMs $N$ we have that
$N \sconf M$ if and only if $L(N) \subseteq L(M')$.
In addition, this result holds even if we restrict attention
to deterministic $N$.
\end{theorem}

\begin{proof}
By Propositions \ref{prop:DFSM_with_sigma} and \ref{prop:langpre_M_cannot_add}
we know that we must have that $L(M') = \trpre (M)$.
The result thus follows from Proposition \ref{prop:langpre_M_not_regular}.
\end{proof}

If $\trpre(M)$ is regular then there is a
corresponding FSM. It would thus be interesting to explore
conditions under which $\trpre(M)$ is guaranteed to be regular
and also properties of $\trpre(M)$ when this is not regular.
There may also be scope to represent
$\trpre(M)$ using an IOTS.

\section{Conformance and multi-tape automata}\label{section:multitape}

While we can decide (in polynomial time) whether $N \wconf M$,
this is quite a weak conformance relation since it does not allow us to bring
together local traces observed at the separate ports.
It seems likely that normally the implementation relation $\sconf$ will be more suitable
and so we consider the problem of deciding whether $N \sconf M$.
In this section we study the problem of deciding language inclusion for
multi-tape automata;
in Section \ref{section:decide_strong}  we use the results proved here
regarding multi-tape automata to show that it is generally undecidable
whether $N \sconf M$ for FSMs $M$ and $N$.
We first define multi-tape FA \cite{rabin59}.

\begin{definition}
An $r$-tape FA with disjoint alphabets $\Sigma_i$, $1 \leq i \leq r$
is a tuple $(S,s_0,\Sigma,h,F)$ in which $S$ is a finite set of states, $s_0 \in S$ is the initial state $F \subseteq S$ is the set of final states and $h : S \times \bigcup_{i=1}^r \Sigma_i  \times S$ is the transition relation.
Then $A$ accepts a tuple $(w_1, \ldots, w_r) \in \Sigma_1^* \times \ldots \times \Sigma_r^*$ if and only if there is some sequence $\sigma \in (\bigcup_{i=1}^r \Sigma_i)^*$ that takes $A$ to a final state such that $\pi_i(\sigma) = w_i$ for all $1 \leq i \leq r$.
We let $\tr(N)$ denote the set of tuples accepted by $N$.
\end{definition}

Given a multi-tape FA $N$ with $r$ tapes,
we will let $L(N)$ denote the language in $(\Sigma_1 \cup \ldots \cup \Sigma_r)^*$
of the corresponding FA.
We obtain $L(N)$ by treating $N$ as a FA with alphabet $\Sigma = \Sigma_1 \cup \ldots \cup \Sigma_r$.

It might seem that deciding whether $N \sconf M$
is similar to deciding whether, for two multi-tape FA $N'$ and $M'$,
the language defined by $N'$ is a subset of that defined by $M'$.
This problem, regarding multi-tape FA, is known to be undecidable \cite{rabin59}.
However, the proof of the result regarding multi-tape FA uses FA in which
not all states are final and in FSMs there is no concept of a state not being a final state.
Thus, the results of \cite{rabin59} are not directly applicable to the problem of deciding whether $N \sconf M$
for FSMs $N$ and $M$ and it appears that the corresponding problem, for multi-tape FA in which all states
are final, has not previously been solved.
In this section we prove that language inclusion is generally undecidable for
multi-tape FA in which all states are final states.
Before we consider decidability issues, 
we investigate the corresponding languages and closure properties.

\begin{proposition}\label{prop:closure}
Let us suppose that $N_1=(S,s_0,h_1,S)$ and $N_2 = (Q,q_0,h_2,Q)$ are multi-tape FA
with the same number of tapes and the same alphabets
and also that all of the states of $N_1$ and $N_2$ are final states.
Then we have the following.
\begin{enumerate}

\item There exists a multi-tape FA $M$ such that $\tr(M) = \tr(N_1) \cup \tr(N_2)$ and all
of the states of $M$ are final states.

\item There exists a multi-tape FA $M$ such that $\tr(M) = \tr(N_1)\tr(N_2)$ and all
of the states of $M$ are final states.

\item There may not exists a multi-tape FA $M$ such that $\tr(M) = \tr(N_1) \setminus \tr(N_2)$ and all
of the states of $M$ are final states.

\item There may not exists a multi-tape FA $M$ such that $\tr(M) = \tr(N_1) \cap \tr(N_2)$ and all
of the states of $M$ are final states.
\end{enumerate}
\end{proposition}

\begin{proof}
We will use $A \oplus B$, for sets $A$ and $B$,
to denote the disjoint union of $A$ and $B$.
For the first result it is sufficient to define the FA $(S \oplus Q \oplus \{r_0\},r_0,h',S \oplus Q \oplus \{r_0\})$
for $r_0 \not \in S \oplus Q$,
in which $h$ is the union of $h_1$ and $h_2$ plus the following transitions:
for every $(s_0,a,s) \in h_1$ we include in $h'$ the tuple $(r_0,a,s)$; and
for every $(q_0,a,q) \in h_2$ we include in $h'$ the tuple $(r_0,a,q)$.

For the second part, it is sufficient to take the disjoint union of $N_1$ and $N_2$ and
for every transition $(s,s',a)$ of $N_1$ add a transition $(s,q_0,a)$.

For the third part,
it is sufficient to observe that for any choice of $N_1$ and $N_2$ we have that
the empty sequence is in $\tr(N_1)$ and $\tr(N_2)$ and so not in
$\tr(N_1) \setminus \tr(N_2)$.
Thus, it is sufficient to choose any $N_1$ and $N_2$ such that $\tr(N_1) \setminus \tr(N_2)$
is non-empty.

For the last part, let us suppose that we have two tapes with alphabets $\{a_1\}$ and $\{a_2\}$,
let $L(N_1) = \{\epsilon, a_1, a_1a_2\}$ and let $L(N_2) = \{\epsilon,a_2,a_2a_1\}$
and so $\tr(N_1) \cap \tr(N_2) = \{\epsilon,a_1a_2,a_2a_1\}$.
\end{proof}

We will use Post's Correspondence Problem to prove that
language inclusion is undecidable.

\begin{definition}
Given sequences $\alpha_1, \ldots, \alpha_m$ and $\beta_1, \ldots, \beta_m$
over an alphabet $\Sigma$,
Post's Correspondence Problem (PCP) is to decide whether there is a sequence
$i_1, \ldots, i_k$ of indices from $[1,m]$  such that
$\alpha_{i_1} \ldots \alpha_{i_k} = \beta_{i_1} \ldots \beta_{i_k}$.
\end{definition}

It is known that Post's Correspondence Problem is undecidable \cite{post46}.

\begin{theorem}\label{theorem:PCP}
Post's Correspondence Problem is undecidable.
\end{theorem}

We now prove the main result from this section.

\begin{theorem}\label{theorem1}
Let us suppose that $N$ and $M$ are multi-tape FA in which all states are final states.
The following problem is undecidable, even when there are only two tapes:
do we have that $\tr(N) \subseteq \tr(M)$?
\end{theorem}

\begin{proof}
We will show that if we can solve this problem then we can also solve
Post's Correspondence Problem.
We therefore assume that we have been given an instance of the PCP
with sequences $\alpha_1, \ldots, \alpha_m$ and $\beta_1, \ldots, \beta_m$ with
alphabet $\Sigma$.
To allow elements of $\Sigma$ to be on both tapes we use two disjoint copies,
$\Sigma_1 = \{f_1(a) | a \in \Sigma\}$ and $\Sigma_2 = \{f_2(a) | a \in \Sigma\}$,
of $\Sigma$.
Given a sequence $x_1 \ldots x_i$ and $j \in \{1,2\}$ we let
$f_j(x_1 \ldots x_i)$ denote $f_j(x_1) \ldots f_j(x_i)$.

We consider multi-tape automata with two tapes with
alphabets $\Sigma_1 \cup \{x, x'\}$ and $\Sigma_2$ respectively,
where $x, x'$ are chosen so that they are not in $\Sigma_1 \cup \Sigma_2$.
We let $N$ denote an FA such that $L(N)$ is the language
that contains all sequences in the regular language
$x((f_1(\alpha_{i_1})f_2(\beta_{i_1}) + \allowbreak \ldots + \allowbreak (f_1(\alpha_{i_k})f_2(\beta_{i_k}))^*x'$ and all prefixes of such sequences.
Clearly, such an FA $N$ exists.
Consider all sequences in $\tr(N)$ that contain an $x$ and an $x'$
and let $\sigma_a$ and $\sigma_b$ be the sequences in the two tapes with $x$ and $x'$ removed.
Then we must have that $\sigma_a$ is of the form $f_1(\alpha_{i_1}) \ldots f_1(\alpha_{i_k})$ and $\sigma_b$
is of the form $f_2(\beta_{i_1}) \ldots f_2(\beta_{i_k})$ with $i_1, \ldots, i_k \in [1,m]$.
In addition, all such combinations correspond to sequences in $L(N)$.
Thus, there is a solution to this instance of the PCP if and only if
$\tr(N)$ contains a tuple $(x\sigma_ax',\sigma_b)$ in which $\sigma_a = \sigma_b$.

The FA $M$ that defines language $L(M)$ is shown in Figure \ref{fig:m} in which
$(a,-)$ denotes elements all of the form $f_1(a)$ and $(-,a)$
denotes all elements of the form $f_2(a)$, $a \in \Sigma$.
Further, $(a,b)$ denotes sequences of the form $f_1(a)f_2(b)$ with $a,b \in \Sigma$ and $a \neq b$.
We use $(a,a)$ to denote sequences  of the form $f_1(a)f_2(a)$ with $a \in \Sigma$.
In all cases where we have sequences with length $2$,
this represents a cycle of length $2$.

\begin{figure}
\begin{center}
\[\UseTips
\entrymodifiers = {}
\xymatrix @+1cm {
*+++[o][F]{s_0} \ar[rd]^{x} & *+++[o][F]{s_1} \ar@(ul,ur)[]^{(-,a)}  \ar[r]^{x'} & *+++[o][F]{s_4} \\
 & *+++[o][F]{s'_0} \ar@(ul,dl)[]_{(a,a)} \ar[u]^{(-,a)} \ar[d]_{(a,-)} \ar[r]^{(a,b)} & *+++[o][F]{s_2} \ar@(ul,ur)^{(a,-),(-,a)}  \ar[r]^{x'} & *+++[o][F]{s_5} \\
& *+++[o][F]{s_3}  \ar@(dl,dr)[]_{(a,-)}   \ar[r]^{x'} & *+++[o][F]{s_5}
}
\]
\caption{Finite Automaton $M$} \label{fig:m}
\end{center}
\end{figure}
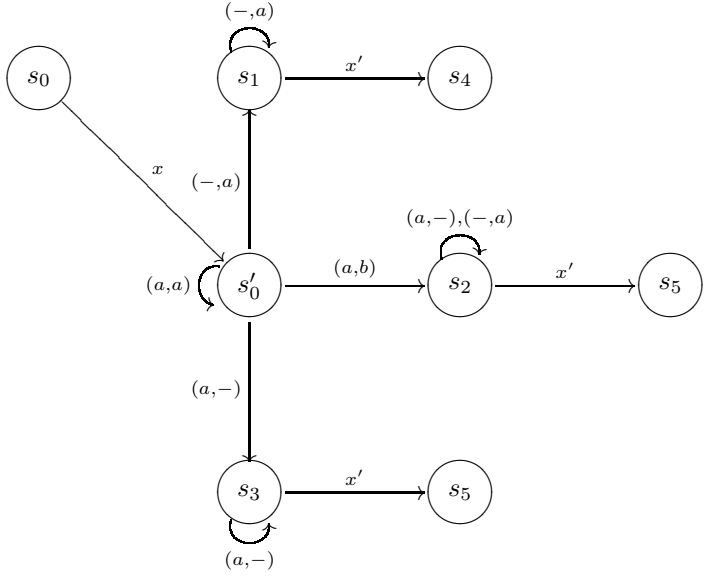

Now consider the problem of deciding whether $\tr(N) \subseteq \tr(M)$.
Specifically, we will focus on the problem of deciding whether $\tr(N)$ is not a subset of $\tr(M)$
and so $\tr(N) \setminus \tr(M)$ is non-empty.
First, note that in $\tr(M)$ we have:
\begin{enumerate}
\item The language defined by paths that pass through $s_1$ contains
all tuples of the form $(xf_1(w_1),f_2(w_2))$ or $(xf_1(w_1)x',f_2(w_2))$
such that $w_1, w_2 \in \Sigma^*$ and $w_1$ is a strict prefix of $w_2$.

\item The language defined by paths that pass through $s_3$ contains all
tuples of the form $(xf_1(w_1),f_2(w_2))$ or $(xf_1(w_1)x',f_2(w_2))$
such that $w_1, w_2 \in \Sigma^*$ and $w_2$ is a strict prefix of $w_1$.

\item The language defined by paths that pass through $s_2$ contains all
tuples of the form
 $(xf_1(w_1w_3),f_2(w_2w_4))$ or $(xf_1(w_1w_3)x',f_2(w_2w_4))$ 
 in which $w_1 \in \Sigma^*$ and $w_2 \in \Sigma^*$ have the same length,
$w_1 \neq w_2$,  and $w_3, w_4 \in \Sigma^*$.

\item The language defined by paths that do not leave $s_0$ contains all tuples of the form
$(xf_1(w),f_2(w))$.
\end{enumerate}

Consider tuples in $\tr(M)$ that do not contain $x'$ and thus tuples of the form
$(xf_1(w_1),f_2(w_2))$ in which $w_1, w_2 \in \Sigma^*$.
Then paths that pass through $s_1$ define all such
tuples in which $w_1$ is a strict prefix of $w_2$ and
paths that pass through $s_3$ define all such
tuples in which $w_2$ is a strict prefix of $w_1$.
In addition, paths that do not leave $s_0$ define
all such tuples in which $w_1 = w_2$ and paths
that pass through $s_2$ define all such tuples in which
$w_1$ and $w_2$ differ after a (possibly empty) common prefix.
Thus, $\tr(M)$ defines all tuples of the form $(xf_1(w_1),f_2(w_2))$ in which $w_1, w_2 \in \Sigma^*$
and so all tuples in $\tr(N)$ that do not contain $x'$ are also in $\tr(M)$.

Now, consider the tuples in $\tr(M)$ that contain $x'$.
These are of the form $(xf_1(w_1)x',f_2(w_2))$ and are
defined by paths that pass through $s_1$, $s_2$ and $s_3$.
By examining the languages defined by paths that pass through
these states we find that
$\tr(M)$ contains the set of tuples of this form in which $w_1 \neq w_2$.
Thus, $\tr(N) \setminus \tr(M)$ is non-empty if and only if
$\tr(N)$ contains a tuple of the form $(xf_1(w)x',f_2(w))$.
But we know that this is the case if and only if there is
a solution to this instance of the PCP and so the result follows
from Theorem \ref{theorem:PCP}.
\end{proof}

For the sake of completeness, we now prove some additional decidability
results regarding multi-tape automata in which all states are final.

\begin{theorem}\label{theorem2}
Let us suppose that $N$ and $M$ are multi-tape FA in which all states are final states.
The following problem is undecidable, even when there are only three tapes:
do we have that $\tr(N) \cap \tr(M)$ contains a non-empty sequence.
\end{theorem}

\begin{proof}
We assume that we have been given an instance of the PCP
with sequences $\alpha_1, \ldots, \alpha_m$ and $\beta_1, \ldots, \beta_m$ with
alphabet $\Sigma$ and follow an approach similar to that used
in the proof of Theorem \ref{theorem1}.
To allow elements of $\Sigma$ on two tapes we use two copies of elements of $\Sigma$ and let
$\Sigma_1 = \{f_1(a) | a \in \Sigma\} \cup \{x\}$, $\Sigma_2 = \{f_2(a) | a \in \Sigma\}$,
and $\Sigma_3 = \{x'\}$.
Given a sequence $a_1, \ldots, a_i \in \Sigma^*$ and $j \in \{1,2\}$ we let
$f_j(a_1 \ldots a_i)$ denote $f_j(a_1) \ldots f_j(a_i)$.

We consider multi-tape automata with three tapes with
alphabets $\Sigma_1$, $\Sigma_2$, and $\Sigma_3$.
We let $N$ denote such an FA such that $L(N)$
contains all tuples formed by sequences in the regular language
$x((f_1(\alpha_{i_1})f_2(\beta_{i_1}) + \allowbreak \ldots + \allowbreak (f_1(\alpha_{i_k})f_2(\beta_{i_k}))^+x'$ and all prefixes of such sequences.
Note that $xx'$ is not contained in $L(N)$.
In addition, we let $M$ denote the multi-tape automaton defined by
the following:
\begin{enumerate}
\item There is a transition from $s_0$ to state $s_1$ and this has label $x'$; 

\item For all $a\in\Sigma$ there is a cycle starting and ending at $s_1$ with label $f_1(a)f_2(a)$;

\item There is a transition from $s_1$ to $s_2$, not involved in the cycles, with label $x$.

\item There are no transition from $s_2$.
\end{enumerate}

Thus, all elements of $\tr(M)$ are either $(\epsilon,\epsilon,\epsilon)$ or contain $x'$.
As a result,
if there is a non-empty element of
$\tr(N) \cap \tr(M)$ then this must contain both $x$ and $x'$.
It is now sufficient to observe that 
this is the case if and only if there is a solution to this instance of the PCP.
\end{proof}

Finally, we prove that equivalence is undecidable for multi-tape FA
in which all states are final states.

\begin{theorem}\label{theorem3}
Let us suppose that $N$ and $M$ are multi-tape FA in which all states are final states.
The following problem is undecidable, even when there are only two tapes:
do we have that $\tr(N) = \tr(M)$?
\end{theorem}

\begin{proof}
First observe that given sets $A$ and $B$ we have that $A \subseteq B$ if
and only if $A \cup B = B$.
Let us suppose that we have two multi-tape automata $N_1$ and $N_2$ with the
same numbers of tapes and the same alphabets and assume that
all states of $N_1$ and $N_2$ are final states.
By Proposition \ref{prop:closure} we know that 
we can construct a multi-tape automaton $N_3$ such that $\tr(N_3) = \tr(N_1) \cup \tr(N_2)$
and all states of $N_3$ are final states.
Thus, if we can decide whether $\tr(N) = \tr(M)$ for two multi-tape FA
that have only final states then we can also decide whether 
$\tr(N_1) \cup \tr(N_2) = \tr(N_2)$
for two multi-tape FA that only have final states.
However, this holds if and only if
$\tr(N_1) \subseteq \tr(N_2)$.
The result thus follows from Theorem \ref{theorem1}.
\end{proof}

We now show how we can represent the problem of deciding language
inclusion for multi-tape FA in terms of deciding whether $N \sconf M$
for suitable FSMs $N$ and $M$.

\section{Deciding Strong Conformance}\label{section:decide_strong}

We have proved some decidability results for multi-tape FA
in which all states are final states.
However, we are interested in FSMs and here a transition has
an input/output pair as a label.
We now show how a multi-tape FA can be represented using an
FSM (with multiple ports)
before using this result to prove that $N \sconf M$ is generally undecidable for
FSMs $N$ and $M$.

In order to extend Theorem \ref{theorem1} to FSMs we 
define a function that takes a multi-port finite automaton and returns an
FSM.

\begin{definition}
Let us suppose that $N = (S, s_0,\Sigma, h, S)$ is a FA with $r$ tapes
with alphabets  $\Sigma_1, \ldots, \Sigma_r$.
We define the FSM $\makeFSM(N)$ with $r+1$ ports as defined below
in which for all $1 \leq p \leq r$ we have that the input alphabet
of $N$ at $p$ is $\Sigma_p$ and the output alphabet is empty and further
we have that the input alphabet at port $r+1$ is empty and the output
alphabet at $r+1$ is $\{0,1\}$.
In the following for $a \in \{0,1\}$ we use $a_{k}$ to denote the $k$-tuple 
whose first $k-1$ elements are empty and whose $k$th element is $a$.

$\makeFSM(N) = (S\cup\{\error\},s_0,\Sigma,\{0_{n+1},1_{n+1}\},h')$ in which
$\error \not \in S$, for all $z \in \Sigma$ we have that
$h'(\error,z) = \{(\error,0_{r+1})\}$ and for all $s \in S$ and $z
\in \Sigma$ we have that $h'(s,z)$ is defined by the following:
\begin{enumerate}
\item If $h(s,z)  = S' \neq \emptyset$ then $h'(s,z) = \{(s',1_{r+1}),(s',0_{r+1}) | s' \in S'\}$;

\item If $h(s,z) = \emptyset$ then $h'(s,z) = \{(\error,0_{r+1})\}$.
\end{enumerate}
\end{definition}

The idea is that while following a path of $N$ the FSM $\makeFSM(N)$ can produce either $0$ or $1$ at port $r+1$ in response to each input but once we diverge from such a path the FSM can then only produce $0$ (at $r+1$) in response to an input.

\begin{lemma}\label{lemma:nfsm}
Let us suppose that $N$ and $M$ are $r$-tape FA with alphabets
$\Sigma_1, \ldots, \Sigma_r$.
Then $\tr(N) \subseteq \tr(M)$ if and only if
$\makeFSM(N) \sconf \makeFSM(M)$.
\end{lemma}

\begin{proof}
First assume that $\makeFSM(N) \sconf \makeFSM(M)$
and we are required to prove that  $\tr(N) \subseteq \tr(M)$.
Assume that $\sigma \in \tr(N)$ and so there exists some
$\sigma' \sim \sigma$ such that $\sigma' \in L(N)$.
Since $\sigma' \in L(N)$ we have that $L(\makeFSM(N))$ contains
the global trace $\rho'$ in which the input portion is
$\sigma'$ and each output is $1_{r+1}$.
Since $\makeFSM(N) \sconf \makeFSM(M)$ we must have that there
is some $\rho'' \in L(\makeFSM(M))$ such that $\rho'' \sim \rho'$.
However, since the outputs are all at port $r+1$ and
the inputs are at ports $1, \ldots, r$ we must have that
$\rho''$ has output portion that contains only $1_{r+1}$ and input portion $\sigma''$
for some $\sigma'' \sim \sigma'$.
Thus, we must have that $\sigma'' \in L(M)$.
Since $\sigma \sim \sigma'$ and $\sigma' \sim \sigma''$
we must have that $\sigma \in \tr(M)$ as required.

Now assume that $\tr(N) \subseteq \tr(M)$
and we are required to prove that  $\makeFSM(N) \sconf \makeFSM(M)$.
Let $\rho$ be some element of $L(\makeFSM(N))$ and it is sufficient
to prove that $\rho \in \tr(\makeFSM(M))$.
Then $\rho = \rho_1\rho_2$ for some maximal $\rho_2$
such that all outputs in $\rho_2$ are $0_{r+1}$.
Let the input portions of $\rho_1$ and $\rho_2$ be
$\sigma_1$ and $\sigma_2$ respectively.
By the maximality of $\rho_2$ we must have that
$\rho_1$ is either empty or ends in output $1_{r+1}$.
Thus, $\sigma_1 \in L(N)$ and so,
since $\tr(N) \subseteq \tr(M)$,
there exists $\sigma'_1 \sim \sigma_1$ with $\sigma'_1 \in L(M)$.
But this means that $M$ can produce the output portion of
$\rho_1$ in response to $\sigma'_1$ and so there exists
$\rho'_1 \in L(\makeFSM(M))$ with $\rho' _1\sim \rho_1$.
By the definition of $\makeFSM(M)$,
since all outputs in $\rho_2$ are $0_{r+1}$ we have
that $\rho'=\rho'_1\rho_2 \in L(\makeFSM(M))$.
The result therefore follows from observing that
$\rho' = \rho'_1\rho_2 \sim \rho_1\rho_2 = \rho$.
\end{proof}

\begin{theorem}\label{theorem:FSM_undecidable}
The following problem is undecidable:
given two multi-port FSMs $N$ and $M$ with the same alphabets,
do we have that $N \sconf M$?
\end{theorem}

\begin{proof}
This follows from Lemma \ref{lemma:nfsm}
and Theorem \ref{theorem1}.
\end{proof}

When considering FSM $M$ with only one port,
we represent the problem of deciding whether two states $s$ and $s'$
of $M$ are equivalent by comparing the FSMs $M_{s}$ and $M_{s'}$,
formed by starting $M$ in $s$ and $s'$ respectively.
However, we also have that the general problem is undecidable.

\begin{theorem}\label{theorem:equiv_states}
The following problem is undecidable:
given a multi-port FSM $M$ and two states $s$ and $s'$ of $M$,
are $s$ and $s'$ equivalent.
\end{theorem}

\begin{proof}
We will prove that we can express the problem of deciding whether
multi-port FSMs are equivalent in terms of state equivalence.
We therefore assume that we have multi-port FSMs $M_1$ and $M_2$
with the same input and output alphabets and we wish to decide whether
$M_1$ an $M_2$ are equivalent.
Let $s_{01}$ and $s_{02}$ denote the initial states of $M_1$ and $M_2$ respectively.
We will construct an FSM $M$ in the following way. We add a new port
$p$ and input $x_p$ at $p$.
The input of $x_p$ in the initial state $s_0$ of $M$ moves $M$ non-deterministically
to either $s_{01}$ or $s_{02}$ and produces no output.
All other input in state $s_0$ moves $M$ to a state $s'_0 \neq s_0$,
that is not a state of $M_1$ or $M_2$,
from which all transitions are self-loops.
The input of $x_p$ in a state of $M_1$ or $M_2$ leads to no output and
no change of state.
Now we can observe that a sequence in the language defined by starting
$M$ in state $s_{0i}$, $i \in \{0,1\}$,
is equivalent under $\sim$ to a sequence from $\tr(M_i)$ followed by
a sequence of zero or more instances of $x_p$.
Thus, $s_{01}$ and $s_{02}$ are equivalent if and only if $M_1$ and $M_2$
are equivalent.
The result thus follows from Theorem \ref{theorem:FSM_undecidable}
and the fact that if we can decide equivalence then we can also decide inclusion.
\end{proof}

We now consider problems relating to distinguishing FSMs and states in testing.
We can only distinguish between FSMs and states on the basis of observations
and each observation, in distributed testing,
defines an equivalence class of $\sim$.

\begin{definition}
It is possible to distinguish FSM $N$ from FSM $M$
in distributed testing if and only if $\tr(N) \not \subseteq \tr(M)$.
Further, it is possible to distinguish between FSMs $N$ and
$M$ in distributed testing if and only if $\tr(N) \not \subseteq \tr(M)$ and
$\tr(M) \not \subseteq \tr(N)$.
\end{definition}

The first part of the definition says that we can only distinguish $N$ from
$M$ in distributed testing if there is some global trace of $N$ that is
not observationally equivalent to a global trace of $M$.
The second part strengthens this by requiring that we can distinguish $N$ from $M$
and also $M$ from $N$.
The following is an immediate consequence of Theorem \ref{theorem:FSM_undecidable}.

\begin{theorem}
The following problems are generally undecidable in distributed testing.
\begin{itemize}
\item Is it possible to distinguish FSM $N$ from FSM $M$?

\item Is it possible to distinguish between FSMs $N$ and $M$?

\end{itemize}
\end{theorem}

Similar to the proof of Theorem \ref{theorem:equiv_states},
we can express the problem of distinguishing between two FSMs as that of
distinguishing two states $s$ and $s'$ of an
FSM $M$.
Thus, the above shows that it is undecidable whether there is some test case that is capable of
distinguishing two states of an FSM or two FSMs.
This complements a previous result \cite{Hierons10_SICOMP},
that it is undecidable whether there is some test case that is
\emph{guaranteed} to distinguish two states or FSMs.

Finally, we give conditions under which equivalence and inclusion are decidable.
The first uses the notion of a Parikh Image of a sequence \cite{Parikh66}.
Given a sequence $\sigma \in \Sigma^*$,
where we have ordered the elements of $\Sigma$ as $a_1, \ldots, a_m$,
the Parikh Image of $\sigma$ is the tuple $(x_1, \ldots, x_m)$
in which for all $1 \leq i \leq m$ we have that $\sigma$ contains $x_i$
instances of $a_i$.
Given a set $A$ of sequences,
the Parikh Image of $A$ is the set of tuples formed by taking the
Parikh Image of each sequence in $A$.

There are classes of languages where the Parikh Image of the language
is guaranteed to be a semi-linear set.
A linear set is defined by a set of vectors $v_0, \ldots, v_k$
that have the same dimension.
Specifically,
the linear set defined by  $v_0, \ldots, v_k$
and is the set of $v_0 + n_1v_1 + \ldots + n_kv_k$ where
$n_1, \ldots, n_k$ are all non-negative integers.
A semi-linear set is a finite union of linear sets.

\begin{proposition}
Let us suppose that multi-port FSMs $N$ and $M$ have the same input
and output alphabets and that for each port $p \in \ports$ we have that
$|\Act_p| \leq 1$.
Then it is decidable whether $N \sconf M$.
\end{proposition}

\begin{proof}
Since for all $p \in \ports$ we have that $|\Act_p| \leq 1$,
for each $\sigma \in \Act^*$ we have that $\sigma$ is
equivalent under $\sim$ to all its permutations.
Thus,
the Parikh Image of a sequence in $L(N)$ or $L(M)$
uniquely defines the corresponding equivalence class.
Thus, $N \sconf M$ if and only if the Parik Image of $L(N)$
is a subset of the Parikh Image of $L(M)$.
However,
these Parikh Images are semi-linear sets and it is decidable whether one
semi-linear set is a subset of another (see, for example,  \cite{KopczynskiT10}).
The result thus follows.
\end{proof}

We now consider the case where each transition produces output
at all ports.

\begin{proposition}
Let us suppose that $M$ is an FSM in which all transitions
produce output at all ports.
Then $N \sconf M$ if and only if $L(N) \subseteq L(M)$.
\end{proposition}

\begin{proof}
First observe that if $N \sconf M$ then each transition of $N$
must also produce output at every port.
Consider a sequence $\sigma \in L(M) \cup L(N)$ that contains $k$ inputs.
Since every transition produces output at all ports,
for a port $p$ we have that $\pi_p(\sigma)$
contains $k$ outputs with each input $x_i$ at $p$ being
between the output produced at $p$ by the previous
input and the output produced at $p$ in response to $x_i$.
Thus, given sequences $\sigma, \sigma' \in L(N) \cup L(M)$
we must have that $\sigma' \sim \sigma$ if and only if
$\sigma' = \sigma$.
The result therefore holds.
\end{proof}

\section{Bounded conformance}\label{section:bounded}

We have seen that it is undecidable whether two FSMs
are related under  $\sconf$.
However,
we might use a weaker notion of conformance where
we only consider sequences of length at most $k$ for some $k$.
This would be relevant when the expected usage does not involve
sequences of length greater than $k$
since, for example,
the system will be reset after at most $k$ inputs.
In this section we define such a weaker implementation relation
and explore the problem of deciding whether two FSMs are
related under this.

First,
we introduce some notation.
We let $IO_k$ denote the set of global traces that have at most $k$ inputs.
In addition, for an FSM $N$ we let $L_k(N) = L(N) \cap IO_k$
denote the set of global traces of $N$ that have at
most $k$ inputs.
We can now define our implementation relation.

\begin{definition}
Given FSMs $N$ and $M$ with the same input and output alphabets,
we say that $N$ strongly $k$-conforms to $M$ if for
all $\sigma \in L_k(N)$ there exists some
$\sigma' \in L(M)$ such that $\sigma' \sim \sigma$.
If this is the case then we write $N \sconf^k M$.
\end{definition}

Clearly, given $N$ and $M$ it is decidable whether $N \sconf^k M$:
we can simply generate every element of $L_k(N)$ and for each
$\sigma \in L_k(N)$ we determine whether $\sigma \in \tr(M)$.
The following shows that this can be achieved in polynomial time
if we have a bound on the number of ports.
We use a result from Mazurkiewicz trace theory.
In Mazurkiewicz trace theory an independence graph is a directed graph where each
vertex of the graph represents an element of $\Act$
and there is an edge between the vertex representing $a \in \Act$
and the vertex representing $b \in \Act$ if and only if
$ab$ and $ba$ are equivalent; $a$ and $b$ are said to be independent \cite{diekert97}.
Thus, for FSMs we have that $a$ and $b$ are independent if and only
if they are at different ports.

\begin{lemma}\label{lemma:oracle}
Given a sequence $\sigma \in IO_k$
and FSM $M$ with $n$ ports,
we can decide whether $\sigma \in \tr(M)$
in time of $O(|\sigma|^n)$.
\end{lemma}

\begin{proof}
The membership problem for a sequence $\sigma$
and rational trace language with alphabet $\Sigma$ and independence relation
$\mathcal{I}$ can be solved in time of $O(|\sigma |^{\alpha})$ where $\alpha$ is the size of the
largest clique in the independence graph \cite{BertoniMS82}.
Since each observation is made at exactly one port
and two observations are independent if and only if they are at different ports,
we have that the maximal cliques of the independence graph all have size $n$
and so $\alpha = n$.
The result therefore follows.
\end{proof}

\begin{theorem}\label{theorem:bounded_poly}
If there are bounds on the value of $k$ and the number $n$ of ports then 
the following problem can be solved in polynomial time:
given FSMs $N$ and $M$ with at most $n$ ports,
do we have that $N \sconf^k M$?
\end{theorem}

\begin{proof}
First observe that the number of elements in $L_k(N)$ is
of $O(q^k)$, where $q$ denotes the maximum number of
transitions leaving a state of $N$.
Thus, since $k$ is bounded,
the elements in $L_k(N)$ can be produced in polynomial time.
It only remains to consider each $\sigma$ in $L_k(N)$
and decide whether it is in $\tr(M)$.
However,
by Lemma \ref{lemma:oracle},
this can be decided in polynomial time.
The result therefore follows.
\end{proof}

Thus, if we have bounds on the number of ports of the system and the
length of sequences we are considering then we can decide whether
$N \sconf^k M$ in polynomial time.
However,
the proof of Theorem \ref{theorem:bounded_poly} introduced
terms that are exponential in $n$ and $k$ and so it is natural
to ask what happens if we do not place bounds on these values.
It transpires that the problem is then NP-hard even for  DFSMs,
the proof using the following.

\begin{definition}
Given boolean variables $z_1, \ldots, z_r$ let $C_1, \ldots, C_k$ denote sets of three literals,
where each literal is either a variable $z_i$ or its negation. The
three-in-one SAT problem is: does there exist an assignment to the
boolean variables such that each $C_j$ contains exactly one true
literal.
\end{definition}

The three-in-one SAT problem is known to be NP-complete \cite{Schaefer78}. 
%

\begin{theorem}
The following problem is NP-hard:
given $k$ and completely specified DFSMs $N$ and $M$,
do we have that $N \sconf^k M$?
\end{theorem}

\begin{proof}
We will show that we can reduce the three-in-one SAT problem to this
and suppose that we have variables $z_1, \ldots, z_r$
and clauses $C_1, \ldots, C_q$. 
We will define a DFSM $M_1$ with $r+q+2$
ports, inputs $z_{-1},z_0,z_1, \ldots, z_r$ at ports $-1,0,1, \ldots, r$ and
outputs $y_1, \ldots, y_{r+q}$ at ports $1, \ldots, r+q$.
We count ports from $-1$ rather than $1$ since the
roles of inputs at $-1$ and $0$ will be rather different from the roles of the
other inputs.

DFSM $M_1$ has four states $s_{0},s_1,s_2,s_3$ in which $s_{0}$ is the
initial state.
The states effectively represent different `modes' and we now describe
the roles of $s_1$ and $s_2$.
In state $s_1$ an input at port 
$p$, $1 \leq p \leq r$, will lead to output at all of the ports corresponding to clauses with literal $z_p$.
In state $s_2$ an input at port 
$p$, $1 \leq p \leq r$, will lead to output at all of the ports corresponding to clauses with literal $\neg z_p$.
The input $z_0$ moves $M_1$ from $s_1$ to $s_2$.
The special input $z_{-1}$ takes $M_1$ from state $s_{0}$ to state $s_1$.

Overall, input $z_0$ does not produce output and only changes the state of $M_1$ if it is in state $s_1$, in which case it takes $M$ to state $s_2$. Input $z_{-1}$ does not produce output and only changes the state of $M_1$ if it is in state $s_{0}$, in which case it takes $M_1$ to state $s_1$.

For an input $z_p$ with $1 \leq p \leq r$ there are four
transitions:
\begin{enumerate}
\item From state $s_1$ there is a transition that,
for all $1 \leq j \leq k$, sends output $y_{r+j}$ to port $r+j$ if $C_j$ contains
literal $z_p$ and otherwise sends no output to port $r+j$.
The transition sends no output to ports $-1, \ldots, r$ and does not change state.

\item From state $s_2$ there is a transition that, for all $1 \leq j \leq k$, sends
output $y_{r+j}$ to port $r+j$ if $C_j$ contains
literal $\neg z_p$ and otherwise sends no output to port $r+j$.
The transition sends no output to ports $-1, \ldots, r$ and does not change state.

\item From state $s_{0}$ there is a transition to state $s_3$ that produces no output.

\item From state $s_3$ there is a transition to state $s_3$ that produces no output.
\end{enumerate}

Now consider the global trace $\sigma$ that starts with input sequence
$z_{-1} z_0 z_1 \ldots z_{r-1}$
and then has input $z_r$ producing the outputs $y_{r+1} \ldots y_{r+q}$;
all outputs are produced in response to the last input.
Clearly we do not have $\sigma \in L(M_1)$.
We now prove that $\sigma \in \tr(M_1)$ if and only if there is a solution to the
instance of the three-in-one SAT problem.
Consider the problem of deciding whether there exists $\sigma' \in L(M_1)$
such that $\sigma' \sim \sigma$.
Clearly the first input in $\sigma'$ must be $z_{-1}$.
Each input $z_p$ is received once by the DFSM and these can be received in any order after $z_{-1}$.
Thus, for all $1 \leq p \leq r$ we do not know whether  $z_p$ will be received before or after $z_0$ in $\sigma'$.
If $z_p$ is received before $z_0$ then an output is sent to all
ports that correspond to clauses that contain literal $z_p$.
If $z_p$ is received after $z_0$ then an output is sent to all
ports that correspond to clauses that contain literal $\neg z_p$.
Thus there exists $\sigma' \in L(M_1)$ such that $\sigma' \sim \sigma$ if
and only if there exists an assignment to the
boolean variables $z_1, \ldots, z_r$ such that each $C_j$ contains exactly one true
literal. 

We now construct DFSMs $N$ and $M$ such that $N \sconf^k M$ if and only if
$\sigma \in \tr(M_1)$.
In the following we assume that $r>1$
and let $\sigma_1$ be the global trace formed from $\sigma$
by replacing the prefix $z_{-1} z_0 z_1$ by $z_1 z_{-1} z_0$.
Thus, $\sigma_1 \sim \sigma$.
We form $N$ from $M_1$ by adding a new path that has label $\sigma_1$.
We add state $s'_3$ such that the input
of $z_1$ in state $s_{0}$ leads to state $s'_3$ (instead of $s_3$)
and no output.
From $s'_3$ we add a transition with label $z_0$ to another new state
$s'_4$.
We repeat this process, adding new states,
until we have a path from $s_0$ with label $z_1 z_0 z_{-1} z_2 z_3 \ldots z_{r-1}$
ending in state $s'_{r+3}$.
We then add a transition from $s'_{r+3}$ to $s'_{r+4}$ with input $z_r$
and the outputs $y_{r+1}, \ldots, y_{r+q}$.
Finally, we complete $N$ by adding a transition to $s_3$ with input $z_p$ and null output
from a state $s'_j$ if there is no transition from $s'_j$ with input $z_p$.
Clearly, $L(N) = L(M_1) \cup \pref(\sigma_1)I^*$.
Let $\sigma_1'$ be defined such that $\sigma_1 = z_{1}\sigma_1'$.
We can similarly form an FSM $M$ from $M_1$ such that
$L(M) = L(M_1) \cup \pref (\{z_{1}\}I\{\sigma_1'\})I^*$.
Since each $I_p$ contains only one input we have that
$\{z_{1}\}I\{\sigma_1'\}I^*$ and $\{\sigma_1\}I^+$ define the
same sets of equivalence classes under $\sim$.
Thus, the equivalence classes of $\pref(\sigma_1)I^*$ and
$\pref (\{z_{1}\}I\{\sigma_1'\})I^*$ under $\sim$ differ only in
the one that contains $\sigma_1$ and we know that $\sigma_1 \sim \sigma$.
We therefore have that 
$N \sconf^k M$, for $k > r+1$,
if and only if $\sigma \in \tr(M_1)$ and we know
that this is the case if and only if the instance of the three-in-one
SAT problem has a solution.
The result follows from the three-in-one SAT problem being NP-hard.
\end{proof}

Naturally,
the results in this section are also relevant when we are looking for tests
of length no longer than $k$ that distinguish states or FSMs.

\section{Conclusions}\label{section:conclusions}

There are important classes of systems
such as cloud systems, web services and wireless sensor networks,
that interact with their environment at physically distributed ports.
In testing such a system we place a local tester at each port
and the local tester (or user) at port $p$ only observes the events
that occur at $p$.
It is known that this reduced observational power,
under which a set of local traces is observed,
can introduce additional controllability and observability problems.

This paper has considered the situation in which there is a
finite state machine (FSM) model $M$ that acts as the specification for
a system that interacts with its environment at physically distributed ports.
We considered the implementation relation $\sconf$
that requires the set of local traces observed to be consistent with
some global trace of the specification.
We investigated the problem of defining a language $\trmax (M)$ such that
we know that $N \sconf M$ if and only if all of the global traces of $N$
are contained in $\trmax (M)$.
We showed that $\trmax (M)$ can be uniquely defined but need not be regular.

We proved that it is generally undecidable whether $N \sconf M$
even if there are only two ports,
although we also gave conditions under which this is decidable.
An additional consequence of this result is that it is undecidable
whether there is a test case (a strategy for each local tester)
that is \emph{capable} of distinguishing two states of an FSM or two FSMs.
This complements earlier results that show that it is undecidable
whether there is a test case that is \emph{guaranteed} to distinguish between
two states of an FSM or two FSMs.
While these results appear to be related the proofs relied on very different
approaches:
the earlier result looked at the problem in terms of multi-player games
while this paper developed and used results regarding
multi-tape automata.

Since it is generally undecidable whether $N \sconf M$ we
defined a weaker implementation relation $\sconf^k$ under
which we only consider input sequences of length $k$ or less.
This is particularly relevant in situations in which it is known
that input sequences of length greater than $k$ need not be
considered since, for example, the system must be reset before this limit has
been reached.
We proved that if we place a bound on $k$ and the number of ports
then we can decide whether $N \sconf^k M$ in polynomial
time but otherwise this problem is NP-hard.

There are several avenues for future work.
First, 
there is the problem of finding weaker conditions under
which we can decide whether $N \sconf M$.
In addition,
it would be interesting to find conditions under which $\trmax (M)$ can
be constructed.
There is also the problem of extending the results to situation in
which we can make additional observations;
for example,
we might consider languages such as CSP in
which we can also observe refusals.
Finally,
one of the motivations for this work was the problem of deciding whether
there is a test case that is capable of distinguishing two states of an FSM
and, despite this being undecidable, it would be interesting to
develop heuristics for this problem.

\bibliographystyle{plain}
\bibliography{trace,test,papers_hierons}

\end{document}